\documentclass[a4paper,twocolumn,10pt,accepted=2024-09-13]{quantumarticle}
\pdfoutput=1
\usepackage[utf8]{inputenc}
\usepackage[english]{babel}
\usepackage[T1]{fontenc}
\usepackage{amsmath}
\usepackage{hyperref}

\usepackage{mathrsfs}
\usepackage{bbm}
\usepackage{amssymb}
\usepackage{amsthm}
\usepackage{mathtools}
\usepackage{bbold}
\usepackage[italicdiff]{physics}

\usepackage{tikz}
\usepackage{lipsum}

\usepackage[numbers,sort&compress]{natbib}

\newtheorem{theorem}{Theorem}
\newtheorem{lemma}[theorem]{Lemma}

\newcommand{\mca}{\mathcal}
\newcommand{\mbb}{\mathbb}

\newcommand{\mket}[1]{|#1\rangle}

\newcommand{\mdyad}[1]{|#1\rangle\langle #1|}

\DeclareMathOperator{\diag}{diag}

\begin{document}

\title{Optimal light cone for macroscopic particle transport in long-range systems: A quantum speed limit approach}

\author{Tan Van Vu}
\email{tan.vu@yukawa.kyoto-u.ac.jp}
\affiliation{Analytical Quantum Complexity RIKEN Hakubi Research Team, RIKEN Center for Quantum Computing (RQC), 2-1 Hirosawa, Wako, Saitama 351-0198, Japan}
\affiliation{Yukawa Institute for Theoretical Physics, Kyoto University, Kitashirakawa Oiwakecho, Sakyo-ku, Kyoto 606-8502, Japan}
\orcid{0000-0001-8184-9433}

\author{Tomotaka Kuwahara}
\email{tomotaka.kuwahara@riken.jp}
\affiliation{Analytical Quantum Complexity RIKEN Hakubi Research Team, RIKEN Center for Quantum Computing (RQC), 2-1 Hirosawa, Wako, Saitama 351-0198, Japan}
\affiliation{PRESTO, Japan Science and Technology (JST), Kawaguchi, Saitama 332-0012, Japan}
\orcid{0000-0002-1612-3940}

\author{Keiji Saito}
\email{keiji.saitoh@scphys.kyoto-u.ac.jp}
\affiliation{Department of Physics, Kyoto University, Kyoto 606-8502, Japan}
\orcid{0000-0002-8464-3198}

\maketitle

\begin{abstract}
Understanding the ultimate rate at which information propagates is a pivotal issue in nonequilibrium physics. Nevertheless, the task of elucidating the propagation speed inherent in quantum bosonic systems presents challenges due to the unbounded nature of their interactions. In this study, we tackle the problem of macroscopic particle transport in a long-range generalization of the lattice Bose-Hubbard model through the lens of the quantum speed limit. By developing a unified approach based on optimal transport theory, we rigorously prove that the minimum time required for macroscopic particle transport is always bounded by the distance between the source and target regions, while retaining its significance even in the thermodynamic limit. Furthermore, we derive an upper bound for the probability of observing a specific number of bosons inside the target region, thereby providing additional insights into the dynamics of particle transport. Our results hold true for arbitrary initial states under both long-range hopping and long-range interactions, thus resolving an open problem of particle transport in generic bosonic systems.
\end{abstract}

\section{Introduction}
The investigation of the velocity at which information propagates is a central topic in quantum mechanics.
On this subject, two primary approaches---the Lieb-Robinson bound \cite{Lieb.1972.CMP} and the quantum speed limit \cite{Mandelstam.1945.JP}---have attracted great attention.
The Lieb-Robinson bound shows the existence of an effective light cone, outside of which information propagation undergoes an exponential decay with distance. Hence, it clarifies the \emph{unreachable} region of information propagation for a given time.
This bound is not only fundamental for establishing causality in physics but also plays a pivotal role in analyzing universal properties in many-body systems \cite{Hastings.2005.PRB,Bravyi.2006.PRL,Hastings.2006.CMP,Nachtergaele.2006.CMP,Osborne.2006.PRL,Hastings.2007.JSM,Acoleyen.2013.PRL,Roberts.2016.PRL,Iyoda.2017.PRL,Haah.2018.JC,Anshu.2021.NP,Alhambra.2021.PRXQ,Kuwahara.2022.PRX,Chen.2023.RPP}, such as the clustering theorem in ground states \cite{Hastings.2006.CMP,Nachtergaele.2006.CMP}, the area law of entanglement \cite{Hastings.2007.JSM,Acoleyen.2013.PRL}, and computational complexity \cite{Haah.2018.JC,Anshu.2021.NP}, to name only a few.
On the other hand, the quantum speed limit addresses the minimum time required for the propagation of information to a \emph{reachable} regime. 
In the celebrated Mandelstam-Tamm speed limit, the minimum time required for the initial density matrix to change to another state is bounded from below by the ratio of the distance between the two states to the energy fluctuation.
To date, the quantum speed limit has been generalized and applied to various fields \cite{Uhlmann.1992.PLA,Margolus.1998.PD,Campo.2013.PRL,Deffner.2013.PRL,Taddei.2013.PRL,Pires.2016.PRX,Mondal.2016.PLA,Deffner.2017.NJP,Shanahan.2018.PRL,Campaioli.2018.PRL,Funo.2019.NJP,GarcaPintos.2019.NJP,Sun.2021.PRL,Vu.2021.PRL,Shiraishi.2021.PRR,Hamazaki.2022.PRXQ,Vu.2023.PRL.TSL,Hasegawa.2023.NC,Deffner.2017.JPA}.
Essentially, these two notions of speed limits, despite having been developed independently in the literature, offer complementary perspectives on the fundamental question of how much time is required to propagate information.
It is thus expected that the proper utilization of these two notions in a unified manner will yield physically novel results and have significant methodological implications, especially for open problems that were previously intractable through the isolated use of either method.

Let us consider the primary problems that Lieb and Robinson have addressed.
In the early days of studies, the Lieb-Robinson bound was explored in the context of short-range interacting spin systems. It was shown that information propagation is effectively confined within the linear light cone, which grows at a finite speed.
In scenarios involving long-range interactions, however, the existence of linear light cones becomes significantly intricate, as such interactions have the capability to instantaneously transmit information to distant places \cite{Eisert.2013.PRL}.
Here, long-range interaction means that the interaction strength between distinct sites obeys a power-law decay as $1/d^\alpha$ with distance $d$.
Intuitively, the shape of the light cone is dependent on the exponent $\alpha$ and may no longer be linear.
Given the ubiquity of long-range systems in nature \cite{Saffman.2010.RMP,Defenu.2023.RMP} and their fundamental interest, it is of great importance to elucidate the shape of the light cone with respect to the power decay $\alpha$.
Thus far, a comprehensive characterization of the shape of the light cone has been established for long-range interacting quantum spin and fermionic systems \cite{Foss-Feig.2015.PRL,Matsuta.2016.AHP,Chen.2019.PRL,Kuwahara.2020.PRX,Else.2020.PRA,Tran.2021.PRX,Kuwahara.2021.PRL2,Tran.2021.PRL, Chen.2021.PRA,Gong.2023.PRL}. 
In Ref.~\cite{Kuwahara.2020.PRX}, the condition for the linear light cone was proven to be $\alpha>2D+1$, where $D$ denotes the spatial dimension.
Furthermore, for the entire regime of $\alpha>2D$, the optimal form of the light cone has been identified as $\tau\gtrsim d^{\min(1,\alpha-2D)}$ \cite{Tran.2021.PRL}.

As another important class, quantum bosonic systems (e.g., the paradigmatic Bose-Hubbard model \cite{Bloch.2008.RMP}) hold significant importance in the study of cold atoms in experimental setups. 
Nevertheless, the theoretical investigation for bosonic systems has persisted due to the unbounded nature of interactions in such systems.
Classifying the effective light cone within the Bose-Hubbard model has been substantiated through both experimental observations and numerical simulations using ultracold gases confined in optical lattices \cite{Luchli.2008.JSM,Carleo.2014.PRA,Cevolani.2015.PRA,Cheneau.2012.N,Takasu.2020.S}.
While the speeds of bosonic transport and information propagation have been studied for bosonic systems with short-range hoppings \cite{Schuch.2011.PRA,Wang.2020.PRXQ,Kuwahara.2021.PRL,Yin.2022.PRX,Kuwahara.2024.NC}, the exploration of long-range cases remains in an early stage of development \cite{Faupin.2022.PRL,Faupin.2022.CMP,Lemm.2023.PRA}.
As the first step toward the Lieb-Robinson bound for long-range cases, the problem of macroscopic bosonic transport has been studied.
Here, macroscopic transport refers to the transfer of a macroscopic number of bosons (i.e., comparable to the total boson number) between separated regions.
Concerning this, Faupin and coworkers \cite{Faupin.2022.PRL} have recently demonstrated the finite velocity of macroscopic particle transport by showing the existence of the linear light cone for $\alpha>D+2$\footnote{Although the range of $\alpha$ considered in Ref.~\cite{Faupin.2022.PRL} is stated as $\alpha\ge D+3$, it can be improved to $\alpha>D+2$ by setting $p=2$ instead of $p=\alpha-D-1$.}. 
However, we are still far from the complete classification of the effective light cone for macroscopic particle transport.
In particular, it is a critical problem to identify the optimal form of the light cone for macroscopic particle transport in bosonic systems, as has been studied for quantum spin and fermionic systems.

\begin{figure}
\centering
\includegraphics[width=1\linewidth]{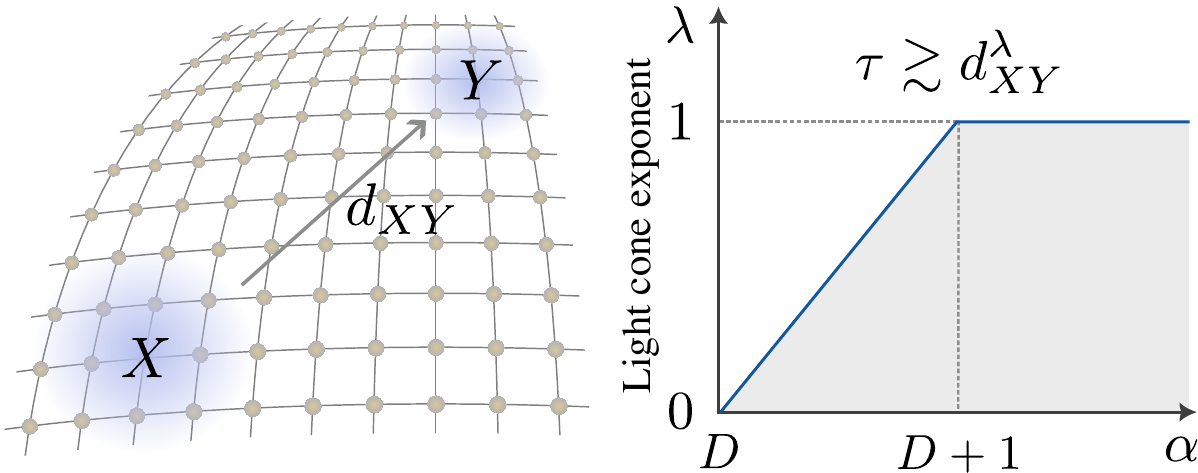}
\protect\caption{(a) Schematic of long-range bosonic systems, wherein bosons are transported from region $X$ to a separate region $Y$ with a distance of $d_{XY}$.
(b) Diagram of our first main result on the minimum operational time $\tau$ for particle transport. The gray area denotes the prohibited region of the exponent.}\label{fig:Cover}
\end{figure}

In this study, we resolve the above problem for macroscopic particle transport with a {\it quantum speed limit approach}. We consider the problem of particle transport in a generalization of the Bose-Hubbard model with {\it long-range hopping} and {\it long-range interactions}, spanning the {\it entire regime} of $\alpha>D$. 
We develop a unified speed limit technique based on optimal transport theory, and we demonstrate the finite speed of macroscopic bosonic transport.
Specifically, we rigorously prove that the minimum time required for the transport of a macroscopic number of bosons is always bounded by the distance $d$ as $\tau\gtrsim d^{\min(1,\alpha-D)}$ [cf.~Eq.~\eqref{eq:speed.limit}; see Fig.~\ref{fig:Cover} for illustration].
The exponent in this bound is optimal, as it can be achieved through a constructive state-transfer protocol. 
Additionally, we derive a limit for the probability of transporting a specific number of bosons into a target region [cf.~Eq.~\eqref{eq:prob.bound}].
We emphasize that all our results unconditionally hold for arbitrary initial states including pure and mixed states.
Notably, our proofs are significantly simple, necessitating only a few fundamental properties of the Wasserstein distance.
These findings highlight the efficacy of our unified approach in overcoming the primary open problems related to speed in quantum many-body physics.

\section{Results}
We consider a generic model of bosons on an arbitrary $D$-dimensional lattice, wherein bosons can hop between arbitrary sites and interact with each other.
The system Hamiltonian is time dependent (i.e., an external control protocol can be applied to the system) and can be expressed in the following form:
\begin{equation}
H_t\coloneqq\sum_{i\neq j\in\Lambda}J_{ij}(t)\hat{b}_i^\dagger \hat{b}_j+\sum_{Z\subseteq\Lambda}h_Z(t).
\end{equation}
Here, $\Lambda$ denotes the set of all the sites in the lattice, $\hat{b}_i^\dagger$ and $\hat{b}_i$ are the bosonic creation and annihilation operators for site $i$, respectively, and $h_Z(t)$ is an arbitrary function of $\{\hat n_i\}_{i\in Z}$, where $\hat n_i\coloneqq \hat{b}_i^\dagger \hat{b}_i$ is the number operator.
Note that $h_Z(t)$ does not need to be local (i.e., $Z$ can be arbitrarily large).
In other words, both {\it long-range hopping} and {\it long-range interactions} are allowed in our setup.
The hopping terms $J_{ij}(t)$ are symmetric [i.e., $J_{ij}(t)=J_{ji}(t)$] and upper bounded by a power law in the Euclidean distance, i.e., $|J_{ij}(t)|\le J/\|i-j\|^\alpha$ for some positive constant $J$, where $\|\cdot\|$ denotes the Euclidean norm and the power decay $\alpha$ satisfies $\alpha>D$\footnote{Note that in the $\alpha \le D$ regime, particle transport speed can be effectively infinite. To elaborate further, there exists a transport protocol that transfers a single particle over a distance $d$ within time $\tau = O(\log d)$ when $\alpha = D$, and $\tau = O(1)$ when $\alpha < D$ \cite{Tran.2020.PRX}. Given our focus on scenarios involving finite bosonic transport speeds, the region where transport speed becomes infinite is excluded from our study.}.
Examples include the standard Bose-Hubbard model, specified by $J_{ij}(t)=0$ for any non-neighboring sites $i$ and $j$ and $\sum_{Z\subseteq\Lambda}h_Z(t)=(U/2)\sum_i\hat n_i(\hat n_i-1)-\mu\sum_i\hat n_i$, where $U$ and $\mu$ are real constants.
Let $\varrho_t$ be the system's density matrix at time $t$.
Then, its time evolution is described by the von Neumann equation ($\hbar=1$ for simplicity),
\begin{equation}
	\dot\varrho_t=-i[H_t,\varrho_t].
\end{equation}
It can be easily shown that the total number of particles in this system is conserved, denoted hereafter by $\mca{N}$.

We introduce some notations that will be used in this study.
For an arbitrary set $X$ of lattice sites, $X_c\coloneqq\Lambda\setminus X$ is the set of sites that are not present in $X$, and $|X|$ denotes its cardinality.
The distance between any two disjoint sets $X$ and $Y$ is defined as the minimum distance between sites belonging to each of them,
\begin{equation}
	d_{XY}\coloneqq\min_{i\in X,\,j\in Y}\|i-j\|.
\end{equation}
The number operator of local bosons occupied inside region $X$ is given by $\hat{n}_X\coloneqq\sum_{i\in X}\hat n_i$, and the corresponding concentration operator is $\overline{n}_X\coloneqq\hat{n}_X/\mca{N}$.
Let $\Pi_{\vec{N}}=\mdyad{\vec{N}}$ be the projection onto the Mott state $\mket{\vec{N}}$ (i.e., $\hat n_i\mket{\vec{N}}=n_i\mket{\vec{N}}$ for $\vec{N}=[n_i]_{i\in\Lambda}$).
Then, the projection $P_{\hat{n}_X\le N_0}$ is defined as
\begin{equation}
	P_{\hat{n}_X\le N_0}\coloneqq\sum_{\vec{N}:\mel{\vec{N}}{\hat{n}_X}{\vec{N}}\le N_0}\Pi_{\vec{N}}.
\end{equation}
The expectation of an observable $A$ with respect to state $\varrho$ is denoted as $\ev{A}_{\varrho}\coloneqq\tr{A\varrho}$.

To fully resolve the open problem of bosonic transport, a novel mathematical framework needs to be developed. Unlike the previous approach based on adiabatic spacetime localization observables in Ref.~\cite{Faupin.2022.PRL}, which could only achieve a partial resolution, we here tackle the problem using a quantum speed limit approach. It is important to note that the quantum speed limit primarily focuses on determining the minimum time required for state changes, but it may not provide accurate predictions for many-body systems \cite{Bukov.2019.PRX}. On the other hand, the Lieb-Robinson bound offers valuable insights into many-body systems but does not furnish information regarding the speed of transitions inside the light cone. Consequently, it seems reasonable to anticipate that unifying these two concepts could yield a novel bound that takes advantage of their respective strengths and establishes new constraints that neither concept can fully encompass on its own. To achieve this, we develop a unified speed limit based on optimal transport theory [cf.~Eq.~\eqref{eq:prim.sl}]. This approach enables us to integrate geometric structures of the underlying dynamics into the speed limit, allowing for meaningful limits to be derived, even when dealing with large-scale systems.

Given the aforementioned setup, we are now ready to explain the results, whose proofs are presented in Appendix \ref{app:method}.
Our first main result is the following statement, which describes a fundamental limit on the operational time required for a macroscopic transport of bosons in long-range systems.
\begin{theorem}\label{thm:tran.velo}
Consider a situation where a fraction $\mu\in(0,1]$ of all bosons is transported from region $X$ to a distant region $Y$.
Then, the operational time $\tau$ required for this macroscopic bosonic transport is lower bounded by the distance between the two regions as
\begin{equation}\label{eq:speed.limit}
	\tau\ge \kappa_{1}^\epsilon d_{XY}^{\min(1,\alpha-D-\epsilon)}.
\end{equation}
Here, $0<\epsilon<\alpha-D$ is an arbitrary number, $\kappa_{1}^\epsilon\coloneqq[J\gamma\zeta(\alpha-\alpha_\epsilon-D+1)]^{-1}\mu$ is a constant independent of the system size, $\alpha_\epsilon\coloneqq\min(1,\alpha-D-\epsilon)$, $\gamma$ is a $O(1)$ constant [cf.~Eq.~\eqref{eq:gamma.def}], and $\zeta$ denotes the Riemann zeta function.
\end{theorem}
We briefly explain the role of the auxiliary variable $\epsilon$.
In simple terms, it is introduced to address the case of $D<\alpha\le D+1$. For $\alpha>D+1$, we can choose $\epsilon=\alpha-D-1$; thus, $\kappa_1^\epsilon=[J\gamma\zeta(\alpha-D)]^{-1}\mu$ and $\min(1,\alpha-D-\epsilon)=1$. 
In the case $D<\alpha\le D+1$, the bound \eqref{eq:speed.limit} explicitly reads
\begin{equation}
	\tau\ge\frac{\mu}{J\gamma\zeta(1+\epsilon)}d_{XY}^{\alpha-D-\epsilon}.
\end{equation}
As seen, $\epsilon$ can be arbitrarily small but should not be zero as the function $\zeta$ can diverge.

Some remarks on Theorem \ref{thm:tran.velo} are in order.
(i) First, macroscopic transport here means that an average number $\mu\mca{N}$ of bosons is transported from $X$ to $Y$, where $\mu$ is a $O(1)$ constant.
Since bosons are indistinguishable, a fraction $\mu$ of the total bosons must be effectively transported from region $X$ to region $Y$ if the following inequality is satisfied:
\begin{equation}\label{eq:mac.tran.def}
	x_Y(\tau)\ge x_{X_c}(0)+\mu,
\end{equation}
where $x_Z(t)$ denotes the boson concentration inside region $Z$ at time $t$.
Inequality \eqref{eq:mac.tran.def} can be interpreted as follows: Even if all bosons outside of $X$ are transported to $Y$, it is still insufficient to achieve the desired concentration; consequently, at least a fraction $\mu$ of the bosons must be effectively transported from $X$ to $Y$.
The result holds for arbitrary initial states, including both pure and mixed states, and for the entire range of the power decay $\alpha>D$.
Notably, this type of time constraint cannot be deduced by the conventional quantum speed limits \cite{Mandelstam.1945.JP,Margolus.1998.PD}.

(ii) Second, bound \eqref{eq:speed.limit} can be physically interpreted as follows.
For the case $\alpha>D+1$, the bound reads $\tau\ge O(d_{XY})$, implying that bosonic transport always takes time at least proportional to the distance.
Surprisingly, this implication is the same as in the case of short-range hopping \cite{Vu.2023.PRL.TSL}, indicating that long-range hopping and long-range interactions do not enhance the speed of macroscopic bosonic transport\footnote{Note that it may not be the case when considering the speed of information propagation.}.
On the other hand, in the case $D<\alpha\le D+1$, the minimum operational time is bounded in terms of the power decay $\alpha$ as $\tau\ge O(d_{XY}^{\alpha-D})$, indicating the effect of long-range hopping on speeding up bosonic transport.

(iii) Last, we discuss the optimality of the bound \eqref{eq:speed.limit} and show that it is optimal.
The optimality here is interpreted in the sense that the power coefficient associated with the distance term in the bound is optimal.
In other words, we need only show that there exist transport protocols such that bosonic transport can be accomplished within time $\tau=O(d_{XY}^{\min(1,\alpha-D)})$.
To this end, we first consider the case of a single particle (i.e., $\mca{N}=1$).
In this case, it was shown that one could construct a rapid state-transfer protocol such that the particle can be transferred between two sites of distance $d$ within time $\tau=O(d)$ for $\alpha>D+1$ and $\tau=O(d^{\alpha-D})$ for $D<\alpha\le D+1$ \cite{Tran.2020.PRX} (see Theorem 11 therein). Therefore, it can be concluded that the bound \eqref{eq:speed.limit} is optimal.
We note that while the form of the optimal light cone is the same as that obtained in Ref.~\cite{Tran.2020.PRX}, the previous study only considers noninteracting particles, which can be simplified to the problem of single free particle. In contrast, our result applies to interacting systems, where boson-boson interactions can be arbitrary and play a crucial role in facilitating bosonic transport (cf.~Lemmas \ref{lem:state.trans2} and \ref{lem:state.trans4}). Next, we show a protocol that saturates the bound for $\alpha>D+1$ in the general case of an extensive boson number (i.e., $\mca{N}\gg 1$). We consider a one-dimensional lattice of length $L$ for simplicity and assume that the initial state is a Mott state $\ket{\mca{N},0,\dots,0}$. Here, $\ket{m}\otimes\ket{n}$ is denoted by $\ket{m,n}$ for convenience. Initially, all bosons are concentrated at site $1$, and we construct a protocol that transports all bosons to site $L$ within time $\tau$. We can sequentially transport bosons within $L-1$ steps, where all bosons are transferred from site $k$ to site $k+1$ at each step $k=1,\dots,L-1$. The transformation in each step can be achieved as $\ket{\mca{N},0}\to\ket{\mca{N}-1,1}\to\ket{1,\mca{N}-1}\to\ket{0,\mca{N}}$. The time consumed in each step is equal to $\pi/(2J)+\pi/(J\sqrt{\mca{N}})=O(1)$ (cf.~Lemmas \ref{lem:state.trans3} and \ref{lem:state.trans4}). Thus, the total time for transporting $\mca{N}$ bosons over a distance of $L-1$ is $O(L)$.

Theorem \ref{thm:tran.velo} concerns the minimum time required to transport an average number of bosons between two regions.
In the sequel, we focus on the probability of finding a number of bosons occupied inside the target region within a finite time.
Our second main result, which is stated below, establishes an upper bound on this probability in terms of the operational time and the distance.
\begin{theorem}\label{thm:prob.bound}
Assume that the initial boson number outside of region $X$ does not exceed $N_0$, i.e., $\ev{P_{\hat{n}_{X_c}\le N_0}}_{\varrho_0}=1$.
Then, the probability of finding bosons inside a disjoint region $Y$ is upper bounded by the operational time $\tau$ and the distance between the two regions as
\begin{equation}\label{eq:prob.bound}
	\ev{P_{\hat{n}_{Y}\ge N_0+\Delta N_0}}_{\varrho_\tau}\le\frac{\kappa_{2}^\epsilon\mca{N}}{\Delta N_0}\frac{\tau}{d_{XY}^{\min(1,\alpha-D-\epsilon)}}.
\end{equation}
Here, $\mca{N}$ is the total number of bosons, $0<\epsilon<\alpha-D$ is an arbitrary number, and $\kappa_{2}^\epsilon\coloneqq J\gamma\zeta(\alpha-\alpha_\epsilon-D+1)$ is a constant independent of the system size.
\end{theorem}
This result holds under a generic setting for both initial states and system dynamics, same as in Theorem \ref{thm:tran.velo}.
It can be interpreted as follows.
Initially, there are at most $N_0$ bosons outside of $X$; therefore, in order to observe more than $N_0$ bosons inside $Y$, a certain number of bosons must be transferred from $X$ to $Y$.
Bound \eqref{eq:prob.bound} indicates that the probability of finding at least $N_0+\Delta N_0$ inside $Y$ is at most linear in time $\tau$ and inversely proportional to boson number $\Delta N_0$ and distance $d_{XY}$.
While the inequality \eqref{eq:prob.bound} is valid for any boson number $\Delta N_0$, it is particularly relevant and meaningful to discuss macroscopic transport [i.e., $\Delta N_0/\mca{N}\eqqcolon\mu$ is a $O(1)$ constant].
In this case, we obtain
\begin{equation}
	\ev{P_{\overline{n}_{Y}\ge\xi}}_{\varrho_\tau}\le\kappa_{2}^\epsilon\mu^{-1}\tau d_{XY}^{-\min(1,\alpha-D-\epsilon)},\label{eq:macro.prob.LB}
\end{equation}
where $\xi\coloneqq (N_0+\Delta N_0)/\mca{N}\ge x_{X_c}(0)+\mu$.
In the case $\alpha>D+2$, a similar probability constraint has been derived for pure states in the absence of long-range interactions \cite{Faupin.2022.PRL} [see Eq.~(7) therein], remaining an open question for the case of $D<\alpha\le D+2$.
Our result thus completely resolves this open problem of bosonic transport for generic long-range systems.

It is worth noting that, at the qualitative level, Theorem \ref{thm:tran.velo} can be viewed as a corollary of Theorem \ref{thm:prob.bound}. For arbitrary $\mu'\in(0,\mu)$, by applying Markov's inequality, we obtain
\begin{align}
	\ev{P_{\overline{n}_{Y}\le x_{X_c}(0)+\mu'}}_{\varrho_\tau}&=\ev{P_{\mbb{1}-\overline{n}_{Y}\ge 1-x_{X_c}(0)-\mu'}}_{\varrho_\tau}\notag\\
	&\le\frac{\ev{\mbb{1}-\overline{n}_{Y}}_{\varrho_\tau}}{1-x_{X_c}(0)-\mu'}\notag\\
	&=\frac{1-x_Y(\tau)}{1-x_{X_c}(0)-\mu'}.
\end{align}
Thus, the probability of finding bosons inside region $Y$ can be lower bounded as
\begin{equation}
	\ev{P_{\overline{n}_{Y}\ge x_{X_c}(0)+\mu'}}_{\varrho_\tau}\ge\frac{x_Y(\tau)-x_{X_c}(0)-\mu'}{1-x_{X_c}(0)-\mu'}.\label{eq:prob.lb}
\end{equation} 
When $x_Y(\tau)\ge x_{X_c}(0)+\mu$, using Eq.~\eqref{eq:prob.lb}, Theorem \ref{thm:prob.bound} can deduce a statement that is qualitatively equivalent to Theorem \ref{thm:tran.velo} as 
\begin{equation}
	\tau\ge\frac{\mu'(\mu-\mu')}{\mu[1-x_{X_c}(0)-\mu']}\kappa_1^\epsilon d_{XY}^{\min(1,\alpha-D-\epsilon)},\label{eq:thm2.cor}
\end{equation}
which can be optimized with $\mu'=1-x_{X_c}(0)-\sqrt{[1-x_{X_c}(0)][1-x_{X_c}(0)-\mu]}$.
While the bound \eqref{eq:thm2.cor} is looser than the bound \eqref{eq:speed.limit} in Theorem \ref{thm:tran.velo}, it captures the same fundamental limit regarding the speed of bosonic transport.

It is worthwhile to investigate the extent to which our results can be extended.
We demonstrate that our main results cannot hold in more general classes of quantum many-boson systems. 
This implies that, in these cases, particle transport at \emph{supersonic} speeds is possible.
Eisert and Gross \cite{Eisert.2009.PRL} have previously established a well-known result that showcases an \emph{exponential} acceleration in quantum communication. 
However, the model used in their study was highly artificial, and there may be more naturally occurring categories of many-boson systems that could potentially exhibit finite-speed quantum transport. 
Our objective here is to illustrate that even by making slight modifications to the Bose-Hubbard model, we can observe an \emph{infinite} acceleration in quantum propagation.

In the sequel, we show that the speed of macroscopic bosonic transport can be supersonic in systems with interaction-induced tunneling terms \cite{Sowinski.2012.PRL}.
These systems have undergone intensive theoretical investigation \cite{Dutta.2015.RPP} and have been observed in experiments \cite{Jurgensen.2014.PRL}.
\begin{theorem}\label{thm:inf.speed}
The speed of particle transport can be infinite in a one-dimensional bosonic system of size $L$ with the time-dependent Hamiltonian given by
\begin{align}
	H_t&=\sum_{\ev{i,j}}J_{ij}(t)\hat{b}_i^\dagger \hat{b}_j+\sum_{\ev{i,j}}T_{ij}(t)\hat{b}_i^\dagger(\hat{n}_i+\hat{n}_j)\hat{b}_j\notag\\
 &+\frac{1}{2}\sum_{i}U_i(t)\hat{n}_i(\hat{n}_i-1)+\sum_{i}\mu_i(t)\hat{n}_i.
 \label{eq:tunneling.Hamiltonian.BH}
\end{align}
Here, $\ev{i,j}$ denotes an ordered pair of neighboring sites (i.e., $|i-j|=1$).
\end{theorem}
We provide a proof outline of Theorem \ref{thm:inf.speed}.
We prove by systematically constructing a protocol that transfers a macroscopic number of bosons to a distant region within a constant of time.
Let the initial state be a Mott state of the form $\ket{\psi_0}=\ket{L,0,\dots,0}$; that is, there are $L$ bosons at site $1$ and no boson at any other site $s\in[2,L]$.
We consider a control protocol that transfers the initial state to the final state $\ket{0,\dots,0,L}$ (i.e., all bosons are transported to site $L$ within time $\tau$).
The protocol is constructed as follows.

At each step $k=1,\dots,L-1$, all the hopping and interaction terms are set to zero, except for sites $k$ and $k+1$.
In other words, the bosonic system effectively becomes a two-site system consisting of sites $k$ and $k+1$.
By performing a subprotocol described below, the state will be transformed as $\ket{L,0}\to\ket{0,L}$.
Consequently, after the final step, the system state $\ket{\psi_\tau}$ will become $\ket{0,\dots,0,L}$ as desired.

The subprotocol that performs the state transformation $\ket{L,0}\to\ket{0,L}$ is composed of three stages, where the first performs $\ket{L,0}\to\ket{L-1,1}$, the second transforms $\ket{L-1,1}\to\ket{1,L-1}$, and the last performs $\ket{1,L-1}\to\ket{0,L}$\footnote{There exists another simple protocol that sequentially transports bosons as $\ket{L,0}\to\ket{L-1,1}\to\ket{L-2,2}\to\dots\to\ket{0,L}$. In this case, the total time required to transport $L$ bosons over a distance of $L-1$ in the thermodynamic limit $L\to\infty$ is given by $\pi^2/(2J)$ (cf.~Lemma \ref{lem:state.trans2}).}.
The Hamiltonian in the first stage is given in the following form:
\begin{align}
	H_1&=J(\hat b_k^\dagger \hat b_{k+1} + \hat b_{k+1}^\dagger \hat b_{k})-U\hat{n}_{k}^2+U(2L-1)\hat{n}_{k}\notag\\
	&+J\qty[\hat b_{k+1}^\dagger(\hat{n}_{k+1}+\hat{n}_{k})\hat b_k+\hat b_{k}^\dagger(\hat{n}_{k+1}+\hat{n}_{k})\hat b_{k+1}].
\end{align}
As proved in Lemma \ref{lem:state.trans2}, the transformation $\ket{L,0}\to\ket{L-1,1}$ can be accomplished within time $\pi/(2JL\sqrt{L})$ by the above Hamiltonian with the $U\to\infty$ limit.
In the second stage, the Hamiltonian reads
\begin{align}
	H_2&=J(\hat b_k^\dagger\hat b_{k+1} + \hat b_{k+1}^\dagger \hat b_{k})\notag\\
	&+J\qty[\hat b_{k+1}^\dagger(\hat{n}_{k+1}+\hat{n}_{k})\hat b_k+\hat b_{k}^\dagger(\hat{n}_{k+1}+\hat{n}_{k})\hat b_{k+1}].
\end{align}
As shown in Lemma \ref{lem:state.trans}, the transformation $\ket{L-1,1}\to\ket{1,L-1}$ can be attained within time $\pi/(2JL)$ by this Hamiltonian.
In the last stage, the Hamiltonian is given by
\begin{align}
	H_3&=J(\hat b_k^\dagger \hat b_{k+1} + \hat b_{k+1}^\dagger \hat b_{k})-U\hat{n}_{k+1}^2+U(2L-1)\hat{n}_{k+1}\notag\\
	&+J\qty[\hat b_{k+1}^\dagger(\hat{n}_{k+1}+\hat{n}_{k})\hat b_k+\hat b_{k}^\dagger(\hat{n}_{k+1}+\hat{n}_{k})\hat b_{k+1}].
\end{align}
As proved in Lemma \ref{lem:state.trans2}, the transformation $\ket{1,L-1}\to\ket{0,L}$ can be achieved within time $\pi/(2JL\sqrt{L})$ by the Hamiltonian $H_3$ in the $U\to\infty$ limit.

For the protocol constructed above, the total time consumed for the bosonic transport is
\begin{equation}
	\tau=\frac{\pi}{2J}\sum_{k=1}^{L-1}\qty(\frac{1}{L}+\frac{2}{L\sqrt{L}})<\frac{\pi}{J}
\end{equation}
for any $L\ge 3$.
As can be observed, macroscopic transport can always be accomplished within a constant of time, independent of the distance between the regions.
Therefore, the speed of macroscopic bosonic transport is not finite in bosonic systems with interaction-induced tunneling terms.

\section{Conclusion}
Using optimal transport theory, we elucidated fundamental limits on the speed of macroscopic particle transport in long-range bosonic systems from the perspective of the quantum speed limit.
The results hold for the entire range $\alpha>D$ of power decay and arbitrary initial states under generic long-range characteristics.
Our findings completely resolve the critical open problem of the speed of macroscopic particle transport.
It is worth noting that the speed of bosonic transport may not be finite in the presence of higher-order hopping terms (i.e., terms like $\hat b_i\hat b_j\hat b_k^\dagger\hat b_l^\dagger$), such as interaction-induced tunneling terms in Eq.~\eqref{eq:tunneling.Hamiltonian.BH}.
We also note that our study does not delve into the exploration of the Lieb-Robinson bound concerning the speed of information propagation. This crucial aspect remains unaddressed and is reserved for future research.

\begin{acknowledgements}
The authors thank Marius Lemm for the helpful communication.
{T.V.V.} was supported by JSPS KAKENHI Grant No.~JP23K13032. {T.K.} acknowledges Hakubi projects of RIKEN and was supported by JST, PRESTO Grant No.~JPMJPR2116. {K.S.} was supported by JSPS KAKENHI Grant No.~JP23K25796.
\end{acknowledgements}

\bibliographystyle{quantum}

\onecolumn
\appendix

\section{Methods}\label{app:method}
\subsection{Unified speed limit based on optimal transport theory}
The key ingredient in the proofs of our results is a novel unified speed limit.
This speed limit is derived using optimal transport theory, which plays crucial roles in various disciplines, such as statistics and machine learning \cite{Kolouri.2017.SPM}, computer vision \cite{Haker.2004.IJCV}, linguistics \cite{Huang.2016.NIPS}, molecular biology \cite{Schiebinger.2019.Cell}, and stochastic thermodynamics \cite{Aurell.2011.PRL,Nakazato.2021.PRR,Dechant.2022.JPA,Vu.2023.PRX}.
Here we describe the unified speed limit and its derivation.

First, we begin with a very general setup, which is applicable to both classical and quantum systems. We consider a time evolution of a vector state $\vb*{x}_t\coloneqq[x_1(t),\dots,x_N(t)]^\top\in\mbb{R}_{\ge 0}^N$ over an undirected graph $\mca{G}(V,E)$ with the vertex set $V=\{1,\dots,N\}$ and edge set $E=\{(i,j)\}$.
For each vertex $i$, let $B_i\coloneqq\{j\,|\,(i,j)\in E\}$ denote the set of neighboring vertices of $i$.
The dynamics of $\vb*{x}_t$ is described by the following deterministic equation:
\begin{equation}\label{eq:det.dyn}
	\dot{x}_i(t)=\sum_{j\in B_i}f_{ij}(t),
\end{equation}
where $f_{ij}(t)=-f_{ji}(t)$ denotes the flow exchange between vertices $i$ and $j$ for $i\neq j$.
Note that the sum $\sum_{i=1}^Nx_i(t)$ is time-invariant.

Next, we introduce the discrete $L^1$-Wasserstein distance between two distributions $\vb*{x}$ and $\vb*{y}$.
Suppose that we have a transport plan that redistributes $\vb*{x}$ to $\vb*{y}$ by sending an amount of $\pi_{ij}\ge 0$ from $x_j$ to $y_i$ with a cost of $c_{ij}\ge 0$ per unit mass for any $i$ and $j$.
Here, $\pi=[\pi_{ij}]\in\mbb{R}_{\ge 0}^{N\times N}$ is a joint probability distribution such that its marginal distributions reduce to $\sum_{j}\pi_{ij}=y_i$ and $\sum_{j}\pi_{ji}=x_i$, defining an admissible transport plan.
The $L^1$-Wasserstein distance is then defined as the minimum transport cost over all feasible plans, given by
\begin{equation}
	\mca{W}(\vb*{x},\vb*{y})\coloneqq\min_{\pi\in\mca{C}(\vb*{x},\vb*{y})}\sum_{i,j}c_{ij}\pi_{ij},
\end{equation}
where $\mca{C}(\vb*{x},\vb*{y})$ denotes the set of couplings $\pi$.
The cost matrix $[c_{ij}]$ can be specified arbitrarily, as long as the following two requirements are fulfilled: (i) symmetry ($c_{ij}=c_{ji}~\forall i,j$) and (ii) the triangle inequality ($c_{ij}+c_{jk}\ge c_{ik} \forall i,j,k$).
The freedom to set up the cost matrix is one of the powerful aspects of the Wasserstein distance. This freedom enables us to include the geometric structure of the underlying dynamics, leading to the derivation of novel physical implications.

Given the setup above, by generalizing the speed limit obtained in Ref.~\cite{Vu.2023.PRL.TSL}, we prove that the minimum time required to transform state $\vb*{x}_0$ into $\vb*{x}_\tau$ is always lower bounded as
\begin{equation}\label{eq:prim.sl}
	\tau\ge\frac{\mca{W}(\vb*{x}_0,\vb*{x}_\tau)}{\ev{\sum_{(i,j)\in E}c_{ij}|f_{ij}(t)|}_\tau},
\end{equation}
where $\ev{z_t}_\tau\coloneqq\tau^{-1}\int_0^\tau\dd{t}z_t$ is the time-average quantity of $z_t$.
The unified speed limit \eqref{eq:prim.sl} includes all the essence to extend the use of the conventional quantum speed limits to a wide range of quantum many-body problems with various geometric structures.
The crucial step in applying this bound effectively is the thoughtful selection of the cost matrix $[c_{ij}]$. With appropriate choices, our new bound \eqref{eq:prim.sl} can play roles similar to the Lieb-Robinson bound.
In the specific case when $c_{ij}=1$ for any $(i,j)\in E$ and $c_{ij}$ represents the shortest-path distance between vertices $i$ and $j$ otherwise, the unified speed limit recovers the topological speed limit \cite{Vu.2023.PRL.TSL}.

The unified speed limit can be derived with the help of the Kantorovich-Rubinstein duality \cite{Villani.2008}, which expresses the Wasserstein distance in an alternative variational form as (see Lemma \ref{lem:KR.dual} for the proof)
\begin{equation}
	\mca{W}(\vb*{x},\vb*{y})=\max_{\vb*{\phi}}\vb*{\phi}^\top(\vb*{x}-\vb*{y}), 
\end{equation}
where the maximum is over all vectors $\vb*{\phi}=[\phi_1,\dots,\phi_N]^\top$ satisfying $|\phi_i-\phi_j|\le c_{ij}~\forall i,j$.
Notice that $x_i(\tau)=x_i(0)+\int_0^\tau\dd{t}\sum_{j\in B_i}f_{ij}(t)$.
Then, by applying the Kantorovich-Rubinstein duality, we can prove Eq.~\eqref{eq:prim.sl} as follows:
\begin{align}
	\mca{W}(\vb*{x}_0,\vb*{x}_\tau)&=\max_{\vb*{\phi}}\vb*{\phi}^\top(\vb*{x}_\tau-\vb*{x}_0)\notag\\
	&=\max_{\vb*{\phi}}{\sum_i\phi_i\int_0^\tau\dd{t}\sum_{j\in B_i}f_{ij}(t)}\notag\\
	&=\max_{\vb*{\phi}}{\sum_{(i,j)\in E}(\phi_i-\phi_j)\int_0^\tau\dd{t}f_{ij}(t)}\notag\\
	&\le\max_{\vb*{\phi}}\sum_{(i,j)\in E}|\phi_i-\phi_j|\int_0^\tau\dd{t}|f_{ij}(t)|\notag\\
	&\le\int_0^\tau\dd{t}\sum_{(i,j)\in E}c_{ij}|f_{ij}(t)|\notag\\
	&=\tau\Big\langle{\sum_{(i,j)\in E}c_{ij}|f_{ij}(t)|}\Big\rangle_\tau.
\end{align}
In what follows, we apply the unified speed limit to prove our main results, Theorems \ref{thm:tran.velo} and \ref{thm:prob.bound}.

\subsection{Proof of Theorem \ref{thm:tran.velo}}
First, we explain some useful notations that will be used later.
For any site $i$, we denote by $i[r]$ the ball of radius $r\ge 0$, which contains all sites $j$ such that $\|i-j\|<r$.
Using this definition, the set $\Lambda\setminus\{i\}$ can be decomposed as
\begin{equation}
	\Lambda\setminus\{i\}=\cup_{\ell=1}^\infty (i[\ell+1]\setminus i[\ell]).
\end{equation}
In addition, we define $\gamma$ as a constant such that the following inequality holds for any $r\ge 0$:
\begin{equation}\label{eq:gamma.def}
	|i[r+1]\setminus i[r]|\le\gamma r^{D-1}.
\end{equation}

We consider a $|\Lambda|$-dimensional vector $\vb*{x}_t$, where $x_i(t)$ denotes the boson concentration occupied at site $i$ [i.e., $x_i(t)\coloneqq\tr{\overline{n}_i\varrho_t}$].
Evidently, $\sum_ix_i(t)=1$ for any $t$.
To define the Wasserstein distance between these vectors, the cost matrix $[c_{ij}]$ needs to be specified. We consider the following cost matrix:
\begin{equation}
	c_{ij}=\|i-j\|^{\alpha_\epsilon},
\end{equation}
where $0<\epsilon<\alpha-D$ is an arbitrary number and $\alpha_\epsilon\coloneqq\min(1,\alpha-D-\epsilon)$ is defined for convenience.
Since $\|i-j\|+\|j-k\|\ge\|i-k\|$ and $a^{z}+b^{z}\ge (a+b)^{z}$ for any $a,b\ge 0$ and $0\le z\le 1$, we can verify that
\begin{equation}
	c_{ij}+c_{jk}\ge c_{ik}
\end{equation}
for any $i,j,k\in\Lambda$.
As can be observed, this distance includes spatial information of the lattice and can be the order of the system size.

We are now ready to prove Theorem \ref{thm:tran.velo}.
Using the relation $[\hat{b}_i,\hat n_i]=\hat{b}_i$, we can show that the time evolution of $x_i(t)$ is governed by the following equation:
\begin{equation}\label{eq:det.eq}
	\dot{x}_i(t)=\sum_{j(\neq i)}2J_{ij}(t)\mca{N}^{-1}\Im\qty[\tr{\hat{b}_i^\dagger \hat{b}_j\varrho_t}]\eqqcolon\sum_{j(\neq i)}\phi_{ij}(t).
\end{equation}
It can be easily verified that $\phi_{ij}(t)=-\phi_{ji}(t)$.
Therefore, the dynamics of $\vb*{x}_t$ can be mapped to Eq.~\eqref{eq:det.dyn}, where the graph $\mca{G}$ is determined by $V=\Lambda$ and $B_i=\Lambda\setminus\{i\}$.
Applying the unified speed limit \eqref{eq:prim.sl} yields
\begin{equation}\label{eq:tmp1}
	\tau\ge\frac{\mca{W}(\vb*{x}_0,\vb*{x}_\tau)}{\ev{\Phi_t}_\tau},
\end{equation}
where we define $\Phi_t\coloneqq (1/2)\sum_{i\neq j}c_{ij}|\phi_{ij}(t)|$.
Since a fraction $\mu$ of bosons must be transported from $X$ to $Y$ [i.e., the inequality \eqref{eq:mac.tran.def} is satisfied], we have
\begin{align}
	\mu\le x_{Y}(\tau)-x_{X_c}(0)&=\sum_{i\in Y,j\in\Lambda}\pi_{ij}-\sum_{i\in\Lambda,j\in X_c}\pi_{ij}\notag\\
	&=\sum_{i\in Y,j\in X}\pi_{ij}-\sum_{i\in Y_c,j\in X_c}\pi_{ij}\notag\\
	&\le \sum_{i\in Y,j\in X}\pi_{ij}.
\end{align}
Therefore,
\begin{align}
	\mca{W}(\vb*{x}_0,\vb*{x}_\tau)\ge \qty(\min_{i\in Y,j\in X}c_{ij})\sum_{i\in Y,j\in X}\pi_{ij}\ge\mu d_{XY}^{\alpha_\epsilon}.\label{eq:tmp2}
\end{align}
On the other hand, by applying the Cauchy-Schwarz inequality, $\phi_{ij}(t)$ can be upper bounded as
\begin{equation}
	|\phi_{ij}(t)|\le |J_{ij}(t)|[{x}_i(t)+{x}_j(t)].
\end{equation}
Consequently, the velocity term $\Phi_t$ is bounded from above as follows:
\begin{align}
	\Phi_t&\le\sum_{i}{x}_i(t)\sum_{j(\neq i)}|J_{ij}(t)|\|i-j\|^{\alpha_\epsilon}\notag\\
	&\le\sum_{i}{x}_i(t)\sum_{\ell=1}^{\infty}\sum_{j\in(i[\ell+1]\setminus i[\ell])}\frac{J}{\|i-j\|^{\alpha-\alpha_\epsilon}}\notag\\
	&\le J\gamma\sum_{i}{x}_i(t)\sum_{\ell=1}^{\infty}\frac{1}{\ell^{\alpha-\alpha_\epsilon-D+1}}\notag\\
	&= J\gamma\zeta(\alpha-\alpha_\epsilon-D+1).\label{eq:tmp3}
\end{align}
Here, we use the facts that $\|i-j\|\ge\ell$ for any $j\in(i[\ell+1]\setminus i[\ell])$ and $|i[\ell+1]\setminus i[\ell]|\le\gamma\ell^{D-1}$ in the third line, and $\zeta(\cdot)$ denotes the Riemann zeta function.
Define $\kappa_{1}^\epsilon\coloneqq[J\gamma\zeta(\alpha-\alpha_\epsilon-D+1)]^{-1}\mu$. Note that $\kappa_{1}^\epsilon$ is independent of the system size and finite because $\alpha-\alpha_\epsilon-D+1\ge 1+\epsilon$.
Combining Eqs.~\eqref{eq:tmp1}, \eqref{eq:tmp2}, and \eqref{eq:tmp3} yields
\begin{equation}
	\tau\ge\kappa_{1}^\epsilon d_{XY}^{\alpha_\epsilon},
\end{equation}
which completes the proof.

\subsection{Proof of Theorem \ref{thm:prob.bound}}
The proof strategy is similar to that of Theorem \ref{thm:tran.velo}.
The different point is that we consider the probability distribution of boson numbers at all sites instead of the vector of average boson concentrations.
Define the probability $p_{\vec{N}}(t)\coloneqq\tr\{\Pi_{\vec{N}}\varrho_t\}$.
Since the total boson number is invariant, we need only consider feasible states $\mket{\vec{N}}$ that satisfy $\tr\{\hat n_{\Lambda}\Pi_{\vec{N}}\}=\mca{N}$.
The time evolution of the probability distribution $\vb*{p}_t=[p_{\vec{N}}(t)]^\top$ can be derived from the von Neumann equation as
\begin{align}
	\dot p_{\vec{N}}(t)&=-i\tr{\Pi_{\vec{N}}[H_t,\varrho_t]}\notag\\
	&=\sum_{i\neq j}iJ_{ij}(t)\sqrt{n_in_j'}\qty(\mel{\vec{N}}{\varrho_t}{\vec{N}'}-\mel{\vec{N}'}{\varrho_t}{\vec{N}})\notag\\
	&\eqqcolon\sum_{i\neq j}\varphi_{\vec{N}\vec{N}'}(t),
\end{align}
where $\vec{N}'$ is obtained from $\vec{N}$ by setting $n_i'=n_{i}-1$, $n_j'=n_{j}+1$, and $n_k'=n_k$ for all $k\neq i,j$.
Note that $\varphi_{\vec{N}\vec{N}'}(t)=-\varphi_{\vec{N}'\vec{N}}(t)\in\mbb{R}$.
Two states $\vec{N}$ and $\vec{M}$ are said to be neighboring if there exist $i\neq j$ such that $|n_i-m_i|=|n_j-m_j|=1$ and $n_k=m_k$ for any $k\neq i,j$.
For such neighboring states, the transport cost is defined as $c_{\vec{N}\vec{N}'}=\|i-j\|^{\alpha_\epsilon}$.
For arbitrary states $\vec{N}$ and $\vec{M}$, the cost is defined as the shortest-path cost over all possible paths connecting the two states,
\begin{equation}
	c_{\vec{M}\vec{N}}=\min\sum_{k=1}^Kc_{\vec{N}_{k}\vec{N}_{k-1}},
\end{equation}
where $\vec{N}_0=\vec{N}$, $\vec{N}_K=\vec{M}$, and $\vec{N}_{k-1}$ and $\vec{N}_k$ are neighboring states for all $1\le k\le K$.
Obviously, $c_{\vec{N}\vec{N}}=0$.
Using the defined shortest-path costs, we consider the following Wasserstein distance:
\begin{equation}
	\mca{W}(\vb*{p},\vb*{q})\coloneqq\min_{\pi\in\mca{C}(\vb*{p},\vb*{q})}\sum_{\vec{M},\vec{N}}c_{\vec{M}\vec{N}}\pi_{\vec{M}\vec{N}}.
\end{equation}
Utilizing the unified speed limit \eqref{eq:prim.sl} for this Wasserstein distance, we can arrive at the following inequality:
\begin{equation}\label{eq:thm2.tmp1}
	\frac{1}{2}\int_0^\tau\dd{t}\sum_{\vec{N}}\sum_{i\neq j}c_{\vec{N}\vec{N}'}|\varphi_{\vec{N}\vec{N}'}(t)|\ge\mca{W}(\vb*{p}_0,\vb*{p}_\tau).
\end{equation}
Furthermore, applying the Cauchy-Schwarz inequality, we obtain
\begin{equation}
	|\varphi_{\vec{N}\vec{N}'}(t)|\le{|J_{ij}(t)|}\qty[n_ip_{\vec{N}}(t)+n_j'p_{\vec{N}'}(t)].
\end{equation}
Consequently, the left-hand side of Eq.~\eqref{eq:thm2.tmp1} can be upper bounded as
\begin{align}
	\frac{1}{2}\sum_{\vec{N}}\sum_{i\neq j}c_{\vec{N}\vec{N}'}|\varphi_{\vec{N}\vec{N}'}(t)|&\le\sum_{\vec{N}}p_{\vec{N}}(t)\sum_{i}n_i\sum_{j(\neq i)}\frac{J}{\|i-j\|^{\alpha-\alpha_\epsilon}}\notag\\
	&\le \sum_{\vec{N}}p_{\vec{N}}(t)\mca{N}J\gamma\zeta(\alpha-\alpha_\epsilon-D+1)\notag\\
	&=\mca{N}J\gamma\zeta(\alpha-\alpha_\epsilon-D+1).\label{eq:thm2.tmp2}
\end{align}
Define $\mca{S}_0=\{\vec{N}\,|\,\sum_{i\in X_c}n_i\le N_0\}$ and $\mca{S}_\tau=\{\vec{N}\,|\,\sum_{i\in Y}n_i\ge N_0+\Delta N_0\}$.
To reach state $\vec{M}\in\mca{S}_\tau$ from state $\vec{N}\in\mca{S}_0$ by hopping bosons between sites, an amount of $\Delta N_0$ bosons must be transferred from sites of $X$ to sites of $Y$.
Let $\vec{N}_{i_0}\to\vec{N}_{i_1}\to\dots\to \vec{N}_{i_L}$ be a sequence of states that transfers a boson from site $i_0\in X$ to site $i_L\in Y$ by iteratively hopping bosons between sites $i_{l-1}$ and $i_l$ for $1\le l\le L$. Then, by applying the triangle inequality, we obtain that the cost of transferring one boson is lower bounded as
\begin{equation}
	\sum_{l=1}^Lc_{\vec{N}_{i_l}\vec{N}_{i_{l-1}}}=\sum_{l=1}^L\|i_l-i_{l-1}\|^{\alpha_\epsilon}\ge\|i_L-i_0\|^{\alpha_\epsilon}\ge d_{XY}^{\alpha_\epsilon}.
\end{equation}
Therefore, we immediately obtain $c_{\vec{M}\vec{N}}\ge\Delta N_0d_{XY}^{\alpha_\epsilon}$ for any $\vec{M}\in\mca{S}_\tau$ and $\vec{N}\in\mca{S}_0$.
Since $\ev{P_{\hat{n}_{X_c}>N_0}}_{\varrho_0}=0$, $p_{\vec{N}}(0)=0$ for any $\vec{N}\notin\mca{S}_0$.
Combining this with $p_{\vec{N}}(0)=\sum_{\vec{M}}\pi_{\vec{M}\vec{N}}$ immediately derives $\pi_{\vec{M}\vec{N}}=0$ for any $\vec{N}\notin\mca{S}_0$ and $\vec{M}$. 
Using these facts, the Wasserstein distance can be lower bounded as follows:
\begin{align}
		\mca{W}(\vb*{p}_0,\vb*{p}_\tau)&=\min_{\pi}\sum_{\vec{M},\vec{N}}c_{\vec{M}\vec{N}}\pi_{\vec{M}\vec{N}}\notag\\
		&\ge \min_\pi\sum_{\vec{M}\in\mca{S}_\tau,\vec{N}\in\mca{S}_0}c_{\vec{M}\vec{N}}\pi_{\vec{M}\vec{N}}\notag\\
		&\ge \Delta N_0d_{XY}^{\alpha_\epsilon}\min_\pi\sum_{\vec{M}\in\mca{S}_\tau}\sum_{\vec{N}\in\mca{S}_0}\pi_{\vec{M}\vec{N}}\notag\\
		&=\Delta N_0d_{XY}^{\alpha_\epsilon}\min_\pi\sum_{\vec{M}\in\mca{S}_\tau}p_{\vec{M}}(\tau)\notag\\
		&=\Delta N_0d_{XY}^{\alpha_\epsilon}\ev{P_{\hat{n}_Y\ge N_0+\Delta N_0}}_{\varrho_\tau}.\label{eq:thm2.tmp3}
\end{align}
Here, we use the relation $\sum_{\vec{N}\in\mca{S}_0}\pi_{\vec{M}\vec{N}}=\sum_{\vec{N}}\pi_{\vec{M}\vec{N}}=p_{\vec{M}}(\tau)$ in the fourth line.
Combining Eqs.~\eqref{eq:thm2.tmp1}, \eqref{eq:thm2.tmp2}, and \eqref{eq:thm2.tmp3} yields
\begin{equation}
	\ev{P_{\hat{n}_Y\ge N_0+\Delta N_0}}_{\varrho_\tau}\le\frac{\mca{N}J\gamma\zeta(\alpha-\alpha_\epsilon-D+1)\tau}{\Delta N_0d_{XY}^{\alpha_\epsilon}},
\end{equation}
which completes the proof by setting $\kappa_{2}^\epsilon\coloneqq J\gamma\zeta(\alpha-\alpha_\epsilon-D+1)$.

\subsection{Useful lemmas}
We formally present useful lemmas about state transformation and their proofs in the following.
\begin{lemma}\label{lem:state.trans}
For any $M\ge 3$, the state transformation $\ket{M-1,1}\leftrightarrow\ket{1,M-1}$ can be accomplished within time $t=\pi/(2JM)$ by the following Hamiltonian:
\begin{equation}
	H=J(\hat b_{1}^\dagger \hat b_{2} + \hat b_{2}^\dagger \hat b_{1})+J\qty[\hat b_{2}^\dagger(\hat{n}_{2}+\hat{n}_{1})\hat b_{1}+\hat b_{1}^\dagger(\hat{n}_{2}+\hat{n}_{1})\hat b_{2}].
\end{equation}
\end{lemma}
\begin{proof}
Let $\ket{\psi_t}=\sum_{m=0}^{M}a_m(t)\ket{m,M-m}$ be the expression of the time-evolved pure state.
Then, the time evolution of $a_m(t)$ can be obtained from the Schr{\"o}dinger equation $\ket{\dot\psi_t}=-iH\ket{\psi_t}$ as
\begin{align}
	\dot a_m(t)&=\braket{m,M-m}{\dot\psi_t}\notag\\
	&=-i\mel{m,M-m}{H}{\psi_t}\notag\\
	&=-iJM\sqrt{m(M-m+1)}a_{m-1}(t)-iJM\sqrt{(m+1)(M-m)}a_{m+1}(t).\label{eq:prop.tmp1}
\end{align}
Defining $\vb*{a}_t\coloneqq[a_0(t),\dots,a_M(t)]^\top$, we can rewrite Eq.~\eqref{eq:prop.tmp1} as
\begin{equation}
	\dot{\vb*{a}}_t=-iJMG\vb*{a}_t,
\end{equation}
where $\vb*{a}_0=[0,\dots,0,1,0]^\top$ and the matrix $G=[g_{mn}]_{0\le m,n\le M}$ is given by
\begin{equation}
	g_{mn}=\begin{cases}
		\sqrt{m(M-m+1)} & \text{if}~n=m-1,\\
		\sqrt{(m+1)(M-m)} & \text{if}~n=m+1,\\
		0 & \text{otherwise}.
	\end{cases}
\end{equation}
Since $\vb*{a}_t=e^{-iJMGt}\vb*{a}_0$, it needs only show that $|[e^{-i\pi G/2}]_{1(M-1)}|=|[e^{-i\pi G/2}]_{(M-1)1}|=1$.
We divide into two cases: $M=2\bar{M}-1$ and $M=2\bar{M}$.
We prove the former, and the latter can be proved analogously.
Define $\tilde{G}=PGP^{-1}=[\tilde{g}_{mn}]_{0\le m,n\le M}$, where $P={\rm diag}(\gamma_0,\dots,\gamma_M)$ is a diagonal matrix that will be determined latter.
We choose $\{\gamma_m\}$ such that
\begin{equation}
	\tilde{g}_{mn}=\begin{cases}
		M-m+1 & \text{if}~n=m-1,\\
		m+1 & \text{if}~n=m+1,\\
		0 & \text{otherwise}.
	\end{cases}
\end{equation}
This can be attained by setting
\begin{equation}
	\frac{\gamma_{m}}{\gamma_{m+1}}=\sqrt{\frac{m+1}{M-m}}.
\end{equation}
The matrix $\tilde{G}$ is nothing but the Sylvester-Kac matrix.
Since $e^{-i\pi G/2}=P^{-1}e^{-i\pi \tilde{G}/2}P$ and $\gamma_{1}=\gamma_{M-1}$, $[e^{-i\pi G/2}]_{1(M-1)}=[e^{-i\pi \tilde{G}/2}]_{1(M-1)}$; thus, we need only prove that $|[e^{-i\pi\tilde{G}/2}]_{1(M-1)}|=1$.
It was shown that $\tilde{G}$ has distinct eigenvalues $\{\lambda_k=M-2k\}_{k=0}^M$ \cite{Kac.1947.AMM}, thus being diagonalizable.
Therefore, its left eigenvectors $\{\ket{v_k^l}\}$ and right eigenvectors $\{\ket{v_k^r}\}$ are orthogonal. That is, $\braket{v_i^l}{v_j^r}=\delta_{ij}$ for any $i$ and $j$. We note that $\braket{v_k^r}$ and $\braket{v_k^l}$ are not necessarily equal to $1$. From the representation $\tilde{G}=\sum_k\lambda_k\dyad{v_k^r}{v_k^l}$, we obtain
\begin{equation}
	e^{-i\pi\tilde{G}/2}=\sum_{k=0}^Me^{-i\pi\lambda_k/2}\dyad{v_k^r}{v_k^l}=i\sum_{k=0}^M(-1)^{\bar{M}+k+1}\dyad{v_k^r}{v_k^l}.\label{eq:prop.tmp7}
\end{equation}
Next, we analytically calculate the right and left eigenvectors.
For each $\lambda_k=M-2k$, from the relation $\tilde{G}\ket{v_k^r}=\lambda_k\ket{v_k^r}$, we obtain
\begin{align}
	v_{k1}^r&=\lambda_kv_{k0}^r,\notag\\
	Mv_{k0}^r+2v_{k2}^r&=\lambda_kv_{k1}^r,\notag\\
	(M-1)v_{k1}^r+3v_{k3}^r&=\lambda_kv_{k2}^r,\notag\\
	&\dots\notag\\
	v_{k(M-1)}^r&=\lambda_kv_{kM}^r.\label{eq:prop.tmp3}
\end{align}
We consider an auxiliary infinite system,
\begin{align}
	v_{k1}^r&=\lambda_kv_{k0}^r,\notag\\
	Mv_{k0}^r+2v_{k2}^r&=\lambda_kv_{k1}^r,\notag\\
	(M-1)v_{k1}^r+3v_{k3}^r&=\lambda_kv_{k2}^r,\notag\\
	&\dots\notag\\
	v_{k(M-1)}^r+(M+1)v_{k(M+1)}^r&=\lambda_kv_{kM}^r,\notag\\
	(M+2)v_{k(M+2)}^r&=\lambda_kv_{k(M+1)}^r,\notag\\
	&\dots\label{eq:prop.tmp4}
\end{align}
If there exists a nontrivial solution of Eq.~\eqref{eq:prop.tmp4} for which $v_{k(M+1)}^r=0$, then we immediately obtain the solution for Eq.~\eqref{eq:prop.tmp3}.
Multiplying all equations of Eq.~\eqref{eq:prop.tmp4} by $1,z,z^2,\dots$ and adding them, we readily obtain
\begin{equation}
	\sum_{m=0}^\infty v_{km}^r(M-m)z^{m+1}+\sum_{m=1}^\infty v_{km}^rmz^{m-1}=\lambda_k\sum_{m=0}^\infty v_{km}^rz^m.\label{eq:prop.tmp5}
\end{equation}
Defining the generating function $f_k(z)\coloneqq\sum_{m=0}^\infty v_{km}^rz^m$, Eq.~\eqref{eq:prop.tmp5} can be translated to the following form:
\begin{equation}
	Mzf_k(z)-z^2f_k'(z)+f_k'(z)=\lambda_kf_k(z),
\end{equation}
which yields the following differential equation:
\begin{equation}
	f_k'(z)=\frac{\lambda_k-Mz}{1-z^2}f_k(z)
\end{equation}
with the initial condition $f_k(0)=v_{k0}^r$.
Solving this differential equation, $f_k(z)$ can be explicitly obtained as
\begin{equation}
	f_k(z)=v_{k0}^r(1+z)^{M-k}(1-z)^{k}.
\end{equation}
Consequently, the components of the right eigenvectors can be calculated as
\begin{equation}
	v_{km}^r=v_{k0}^r\sum_{i=m-k}^{M-k}(-1)^{m-i}{k \choose m-i}{M-k \choose i}.\label{eq:prop.right}
\end{equation}
We now calculate the left eigenvectors.
From the relation $\sum_{k}{v_{ki}^r}{v_{km}^l}=\delta_{im}$, we can proceed as follows:
\begin{align}
	z^{m}&=\sum_{i=0}^Mz^{i}\delta_{im}\notag\\
	&=\sum_{i=0}^Mz^{i}\sum_{k=0}^Mv_{ki}^rv_{km}^l\notag\\
	&=\sum_{k=0}^Mv_{km}^l\sum_{i=0}^Mv_{ki}^rz^{i}\notag\\
	&=\sum_{k=0}^Mv_{km}^lv_{k0}^r(1+z)^{M-k}(1-z)^{k},
\end{align}
from which the following relation can be obtained:
\begin{equation}\label{eq:prop.tmp6}
	\frac{z^m}{(1+z)^M}=\sum_{k=0}^Mv_{km}^lv_{k0}^r\qty(\frac{1-z}{1+z})^{k}.
\end{equation}
Let $\xi=(1-z)/(1+z)$, Eq.~\eqref{eq:prop.tmp6} can be rewritten as
\begin{equation}
	\sum_{k=0}^Mv_{km}^lv_{k0}^r\xi^k=2^{-M}(1-\xi)^m(1+\xi)^{M-m}.
\end{equation}
By comparing the coefficients of the terms $\xi^k$, the components of the left eigenvectors can be derived as
\begin{equation}
	v_{km}^l=\frac{2^{-M}}{v_{k0}^r}\sum_{i=k-m}^{M-m}(-1)^{k-i}{m \choose k-i}{M-m \choose i}.\label{eq:prop.left}
\end{equation}
Inserting the expressions obtained in Eqs.~\eqref{eq:prop.right} and \eqref{eq:prop.left} to Eq.~\eqref{eq:prop.tmp7}, we can calculate $[e^{-i\pi\tilde{G}/2}]_{1(M-1)}$ as follows:
\begin{align}
	[e^{-i\pi\tilde{G}/2}]_{1(M-1)}&=i\sum_{k=0}^M(-1)^{\bar{M}+k+1}\qty[\dyad{v_k^r}{v_k^l}]_{1(M-1)}\notag\\
	&=i\sum_{k=0}^M(-1)^{\bar{M}+k+1}v_{k1}^rv_{k(M-1)}^l\notag\\
	&=2^{-M}(-1)^{\bar{M}+1}i\qty[M+\sum_{k=0}^{M-1}\frac{(M-2k)^2}{M-k} {M-1 \choose k}]\notag\\
	&=(-1)^{\bar{M}+1}i.
\end{align}
Here, we use the following facts to obtain the last line:
\begin{align}
	\sum_{k=0}^{M-1}{M\choose k}&=2^{M}-1,\\
	\sum_{k=0}^{M-1}k{M\choose k}&=M2^{M-1}-M,\\
	\sum_{k=0}^{M-1}k(k-1){M\choose k}&=M(M-1)2^{M-2}-M(M-1).
\end{align}
Therefore, $|[e^{-i\pi\tilde{G}/2}]_{1(M-1)}|=1$, which completes the proof.
\end{proof}

\begin{lemma}\label{lem:state.trans2}
For any $M\ge 3$ and $1\le k\le M$, the state transformation $\ket{k,M-k}\leftrightarrow\ket{k-1,M-k+1}$ can be accomplished within time $t=\pi/[2JM\sqrt{k(M-k+1)}]$ in the $U\to+\infty$ limit by the following Hamiltonian:
\begin{align}
	H&=J(\hat b_1^\dagger \hat b_{2} + \hat b_{2}^\dagger \hat b_{1})+J\qty[\hat b_{2}^\dagger(\hat{n}_{2}+\hat{n}_{1})\hat b_1+\hat b_{1}^\dagger(\hat{n}_{2}+\hat{n}_{1})\hat b_{2}]\notag\\
	&-U\hat{n}_{2}^2+U(2M-2k+1)\hat{n}_{2}.
\end{align}
\end{lemma}
\begin{proof}
Following the same approach in Lemma \ref{lem:state.trans}, we can derive the following differential equation:
\begin{equation}
	\dot{\vb*{a}}_t=-iG\vb*{a}_t,
\end{equation}
where the matrix $G=[g_{mn}]_{0\le m,n\le M}$ is given by
\begin{equation}
	g_{mn}=\begin{cases}
		JM\sqrt{m(M-m+1)} & \text{if}~n=m-1,\\
		JM\sqrt{(m+1)(M-m)} & \text{if}~n=m+1,\\
		U(M-m)(M+m-2k+1) & \text{if}~n=m,\\
		0 & \text{otherwise}.
	\end{cases}
\end{equation}
The real symmetric matrix $G$ has distinct real eigenvalues $\lambda_0,\lambda_1,\dots,\lambda_M$, and their corresponding eigenvectors are denoted by $\{\ket{v_m^r}\}$.
In the $U\to+\infty$ limit, we can easily show that
\begin{align*}
	\lim_{U\to+\infty}\qty[\lambda_{k-1}-U(M-k+1)(M-k)]&=JM\sqrt{k(M-k+1)},\\
	\lim_{U\to+\infty}\qty[\lambda_{k}-U(M-k)(M-k+1)]&=-JM\sqrt{k(M-k+1)},\\
	\lim_{U\to+\infty}\qty[\lambda_m-U(M-m)(M+m-2k+1)]&=0~\forall m\neq k-1,k,\\
	\lim_{U\to+\infty}\ket{v_{k-1}^r}&=[\underbrace{0,\dots,0}_{k-1},1/\sqrt{2},1/\sqrt{2},0,\dots,0]^\top,\\
	\lim_{U\to+\infty}\ket{v_k^r}&=[\underbrace{0,\dots,0}_{k-1},1/\sqrt{2},-1/\sqrt{2},0,\dots,0]^\top,\\
	\lim_{U\to+\infty}\ket{v_m^r}&=[\delta_{mi}]^\top~\forall m\neq k-1,k.
\end{align*}
Defining $P=(\ket{v_0^r},\dots,\ket{v_M^r})$ and noting that $P^\top=P$ and $P^2=\mbb{1}$, we obtain $G=P\diag{\qty(\lambda_0,\dots,\lambda_M)}P$.
Consequently, the matrix $e^{-iGt}$ can be calculated in the $U\to+\infty$ limit as $e^{-iGt}=P\diag{\qty(e^{-i\lambda_0t},\dots,e^{-i\lambda_Mt})}P$.
We need only show that $|[e^{-iGt}]_{k(k-1)}|=|[e^{-iGt}]_{(k-1)k}|=1$ for $t=\pi/[2JM\sqrt{k(M-k+1)}]$.
Indeed, we can prove as follows:
\begin{align}
	|[e^{-iGt}]_{(k-1)k}|&=\qty|e^{-i\lambda_{k-1}t}-e^{-i\lambda_{k}t}|/2\notag\\
	&=\qty|1-e^{i(\lambda_{k-1}-\lambda_{k})t}|/2\notag\\
	&=\qty|1-e^{2iJM\sqrt{k(M-k+1)}t}|/2\notag\\
	&=\qty|1-e^{i\pi}|/2\notag\\
	&=1.
\end{align}
This completes the proof.
\end{proof}

\begin{lemma}\label{lem:state.trans3}
For any $M\ge 3$, the state transformation $\ket{M-1,1}\leftrightarrow\ket{1,M-1}$ can be accomplished within time $t=\pi/(2J)$ by the following Hamiltonian:
\begin{equation}
	H=J(\hat b_{1}^\dagger \hat b_{2} + \hat b_{2}^\dagger \hat b_{1}).
\end{equation}
\end{lemma}
\begin{proof}
The proof is analogous to that of Lemma \ref{lem:state.trans}.
\end{proof}
\begin{lemma}\label{lem:state.trans4}
For any $M\ge 3$ and $1\le k\le M$, the state transformation $\ket{k,M-k}\leftrightarrow\ket{k-1,M-k+1}$ can be accomplished within time $t=\pi/[2J\sqrt{k(M-k+1)}]$ in the $U\to+\infty$ limit by the following Hamiltonian:
\begin{align}
	H&=J(\hat b_1^\dagger \hat b_{2} + \hat b_{2}^\dagger \hat b_{1})-U\hat{n}_{2}^2+U(2M-2k+1)\hat{n}_{2}.
\end{align}
\end{lemma}
\begin{proof}
The proof is analogous to that of Lemma \ref{lem:state.trans2}.
\end{proof}

\begin{lemma}\label{lem:KR.dual}
The Wasserstein distance $\mca{W}(\vb*{x},\vb*{y})\coloneqq\min_{\pi\in\mca{C}(\vb*{x},\vb*{y})}\sum_{i,j}\pi_{ij}c_{ij}$ can be expressed in the following form:
\begin{equation}
	\mca{W}(\vb*{x},\vb*{y})=\max_{\vb*{\phi}}\vb*{\phi}^\top(\vb*{x}-\vb*{y}),\label{eq:lem.dual}
\end{equation}
where the maximum is over all vectors $\vb*{\phi}=[\phi_1,\dots,\phi_N]^\top$ satisfying $|\phi_i-\phi_j|\le c_{ij}~\forall i,j$.
\end{lemma}
\begin{proof}
The constraint $\pi\in\mca{C}(\vb*{x},\vb*{y})$ on $\pi$ in the definition of $\mca{W}(\vb*{x},\vb*{y})$ can be relaxed to $\pi\in\mbb{R}_{\ge 0}^{N\times N}$ by adding a term $R(\pi)$ to the variational expression such that $R(\pi)=0$ for $\pi\in\mca{C}(\vb*{x},\vb*{y})$ and $R(\pi)=+\infty$ otherwise.
Although there are infinitely many ways to define $R(\pi)$, we conveniently choose
\begin{equation}
	R(\pi)=\max_{\vb*{f},\vb*{g}\in\mbb{R}^N}\qty[\vb*{f}^\top\vb*{x}+\vb*{g}^\top\vb*{y}-\sum_{i,j}(g_i+f_j)\pi_{ij}].
\end{equation}
Consequently, the Wasserstein distance can be expressed as
\begin{align}
	\mca{W}(\vb*{x},\vb*{y})&=\min_{\pi\in\mbb{R}_{\ge 0}^{N\times N}}\qty[\sum_{i,j}\pi_{ij}c_{ij}+R(\pi)]\notag\\
	&=\min_{\pi\in\mbb{R}_{\ge 0}^{N\times N}}\max_{\vb*{f},\vb*{g}\in\mbb{R}^N}\qty[\sum_{i,j}\pi_{ij}(c_{ij}-g_i-f_j)+\vb*{f}^\top\vb*{x}+\vb*{g}^\top\vb*{y}]\notag\\
	&=\min_{\pi\in\mbb{R}_{\ge 0}^{N\times N}}\max_{\vb*{f},\vb*{g}\in\mbb{R}^N}L(\pi,\vb*{f},\vb*{g}).\label{eq:lem.KR.tmp0}
\end{align}
Here, we define $L(\pi,\vb*{f},\vb*{g})\coloneqq\sum_{i,j}\pi_{ij}(c_{ij}-g_i-f_j)+\vb*{f}^\top\vb*{x}+\vb*{g}^\top\vb*{y}$.
Since $L(\pi,\vb*{f},\vb*{g})$ is a linear function in its arguments, the max-min inequality is also an equality, according to von Neumann's minimax theorem.
Therefore, we can proceed further from Eq.~\eqref{eq:lem.KR.tmp0} as
\begin{align}
	\mca{W}(\vb*{x},\vb*{y})&=\min_{\pi\in\mbb{R}_{\ge 0}^{N\times N}}\max_{\vb*{f},\vb*{g}\in\mbb{R}^N}L(\pi,\vb*{f},\vb*{g})\notag\\
	&=\max_{\vb*{f},\vb*{g}\in\mbb{R}^N}\min_{\pi\in\mbb{R}_{\ge 0}^{N\times N}} L(\pi,\vb*{f},\vb*{g})\notag\\
	&=\max_{\vb*{f},\vb*{g}\in\mbb{R}^N}\qty[\vb*{f}^\top\vb*{x}+\vb*{g}^\top\vb*{y}+\min_{\pi\in\mbb{R}_{\ge 0}^{N\times N}}\sum_{i,j}\pi_{ij}(c_{ij}-g_i-f_j)].\label{eq:lem.KR.tmp1}
\end{align}
If there exist $i$ and $j$ such that $c_{ij}-g_i-f_j<0$, then $\min_{\pi\in\mbb{R}_{\ge 0}^{N\times N}}\sum_{i,j}\pi_{ij}(c_{ij}-g_i-f_j)=-\infty$.
Therefore, we need only consider $\vb*{f},\vb*{g}\in\mbb{R}^N$ such that $c_{ij}-g_i-f_j\ge 0$, or equivalently $f_j+g_i\le c_{ij}~\forall i,j$.
Therefore, Eq.~\eqref{eq:lem.KR.tmp1} can be expressed as
\begin{align}
	\mca{W}(\vb*{x},\vb*{y})&=\max_{\vb*{f},\vb*{g}\in\mbb{R}^N:f_j+g_i\le c_{ij}}\qty[\vb*{f}^\top\vb*{x}+\vb*{g}^\top\vb*{y}+\min_{\pi\in\mbb{R}_{\ge 0}^{N\times N}}\sum_{i,j}\pi_{ij}(c_{ij}-g_i-f_j)].\label{eq:lem.KR.tmp2}
\end{align}
Now, noticing that $\min_{\pi\in\mbb{R}_{\ge 0}^{N\times N}}\sum_{i,j}\pi_{ij}(c_{ij}-g_i-f_j)=0$ because $c_{ij}-g_i-f_j\ge 0$, we immediately obtain the following equality from Eq.~\eqref{eq:lem.KR.tmp2}:
\begin{equation}
	\mca{W}(\vb*{x},\vb*{y})=\max_{\vb*{f},\vb*{g}\in\mbb{R}^N:f_j+g_i\le c_{ij}}(\vb*{f}^\top\vb*{x}+\vb*{g}^\top\vb*{y}).
\end{equation}
For any fixed $\vb*{g}$, the optimal $\tilde{\vb*{f}}$ that maximizes $(\vb*{f}^\top\vb*{x}+\vb*{g}^\top\vb*{y})$ under the constraints is uniquely determined as
\begin{equation}
	\tilde{f}_i=\min_{j}(c_{ij}-g_j)~\forall i.
\end{equation}
Since $c_{ij}\le c_{kj}+c_{ik}~\forall k$, we obtain
\begin{equation}
	\tilde{f}_i\le\min_{j}(c_{kj}-g_j+c_{ik})=\tilde{f}_k+c_{ik}.
\end{equation}
This implies $|\tilde{f}_i-\tilde{f}_k|\le c_{ik}$ as $c_{ik}=c_{ki}$, which also means that the vector $\tilde{\vb*{f}}$ satisfies the 1-Lipschitz condition.
Moreover, since $c_{ii}=0$, we also have
\begin{equation}
	\tilde{f}_i\le-g_i~\text{or}~g_i\le -\tilde{f}_i.
\end{equation}
Combining these inequalities, we can upper bound the Wasserstein distance as
\begin{align}
	\mca{W}(\vb*{x},\vb*{y})&=\max_{\vb*{g}\in\mbb{R}^N}(\tilde{\vb*{f}}^\top\vb*{x}+\vb*{g}^\top\vb*{y})\notag\\
	&\le\max_{\vb*{g}\in\mbb{R}^N}(\tilde{\vb*{f}}^\top\vb*{x}-\tilde{\vb*{f}}^\top\vb*{y})\notag\\
	&\le\max_{\vb*{\phi}\in\mbb{R}^N:|\phi_i-\phi_j|\le c_{ij}}\vb*{\phi}^\top(\vb*{x}-\vb*{y}).\label{eq:lem.W.ub}
\end{align}
Let $\tilde{\pi}$ be the optimal coupling that attains the Wasserstein distance, that is, $\mca{W}(\vb*{x},\vb*{y})=\sum_{i,j}\tilde{\pi}_{ij}c_{ij}$.
Then, we can also prove that
\begin{align}
	\max_{\vb*{\phi}\in\mbb{R}^N:|\phi_i-\phi_j|\le c_{ij}}\vb*{\phi}^\top(\vb*{x}-\vb*{y})&=\max_{\vb*{\phi}\in\mbb{R}^N:|\phi_i-\phi_j|\le c_{ij}}\sum_{i,j}(\phi_i\tilde{\pi}_{ji}-\phi_i\tilde{\pi}_{ij}),\notag\\
	&=\max_{\vb*{\phi}\in\mbb{R}^N:|\phi_i-\phi_j|\le c_{ij}}\sum_{i,j}\tilde{\pi}_{ij}(\phi_j-\phi_i),\notag\\
	&\le\sum_{i,j}\tilde{\pi}_{ij}c_{ij}\notag\\
	&=\mca{W}(\vb*{x},\vb*{y}).\label{eq:lem.W.lb}
\end{align}
From Eqs.~\eqref{eq:lem.W.ub} and \eqref{eq:lem.W.lb}, we immediately obtain the desired result \eqref{eq:lem.dual}.
\end{proof}

\end{document}